\documentclass[submission,copyright,creativecommons]{eptcs}
\providecommand{\event}{LSFA 2012} 

\usepackage{breakurl}             
\usepackage{amssymb}
\usepackage{amsmath}
\usepackage{mathrsfs}
\usepackage{stmaryrd}
\usepackage{amsthm}
\usepackage{xifthen} 
\usepackage{rotating}
\usepackage{enumerate}
\usepackage{enumitem}
\usepackage{rotating}

\newtheorem{theorem}{Theorem}[section]
\newtheorem{corollary}[theorem]{Corollary}
\newtheorem{proposition}[theorem]{Proposition}
\newtheorem{lemma}[theorem]{Lemma}

\newtheorem{definition}{Definition}[section]

\newcommand{\bA}{\mathbf{A}}
\newcommand{\bB}{\mathbf{B}}
\newcommand{\bC}{\mathbf{C}}

\newcommand{\bI}{\mathbf{I}}

\newcommand{\bM}{\mathbf{M}}

\newcommand{\bP}{\mathbf{P}}

\newcommand{\bS}{\mathbf{S}}

\newcommand{\ga}{\alpha}
\newcommand{\gb}{\beta}
\newcommand{\gc}{\gamma}
\newcommand{\gd}{\delta}

\newcommand{\gl}{\lambda}

\newcommand{\gn}{\nu}

\newcommand{\gp}{\pi}

\newcommand{\gs}{\sigma}

\newcommand{\go}{\omega}

\newcommand{\gL}{\Lambda}

\newcommand{\rf}{\mathrm{f}}

\newcommand{\cA}{\mathcal{A}}
\newcommand{\cB}{\mathcal{B}}
\newcommand{\cC}{\mathcal{C}}

\newcommand{\cH}{\mathcal{H}}

\newcommand{\cL}{\mathcal{L}}

\newcommand{\cP}{\mathcal{P}}

\newcommand{\cS}{\mathcal{S}}

\renewcommand{\iff}{\Leftrightarrow}

\newcommand\ssk{{\bf k}}
\newcommand\sss{{\bf s}}

\newcommand{\gLo}{\Lambda^{\mathrm{o}}} 
\newcommand{\Var}{\mathrm{Var}}  

\newcommand{\Con}{\mathrm{Con}} 

\newcommand{\CPO}{\mathbf{CPO}}
\newcommand{\SD}{\mathbf{SD}} 

\newcommand{\Inf}{\mathbf{Inf}} 

\newcommand{\EATS}{\mbox{\sf EATS}} 

\newcommand{\dom}{\mathrm{do}} 
\newcommand{\fst}{\mathrm{fst}} 
\newcommand{\snd}{\mathrm{snd}} 

\newcommand{\metal}{\mathrel{\lambda\!\!\!\!\!\lambda}\!\!\!}
\newcommand{\Int}[1]{\llbracket #1\rrbracket} 
\newcommand{\st}{:}

\newcommand{\cldn}[1]{{#1}\!\downarrow}

\newcommand{\pow}[1]{{#1}} 
\newcommand{\pro}[1]{\widetilde{#1}} 

\newcommand{\expsp}[2]{#1 \Rightarrow #2} 
\newcommand{\judge}[4]{
	\ifthenelse{\isempty{#4}}
	{
	#1 \rhd_{#2} #3 
	}
	{
	#1 \rhd_{#2} #3 : #4
	}
}

\newcommand{\comp}{\mbox{$\circ$}}

\newcommand{\id}{\mathrm{id}} 

\newcommand{\Th}{{\mbox{Eq}}}

\input prooftree.sty

\title{Minimal lambda-theories by ultraproducts}

\author{
Antonio Bucciarelli \qquad Alberto Carraro
\institute{PPS, Universit\'{e} Denis Diderot Paris, France}
\email{\{acarraro,buccia\}@pps.univ-paris-diderot.fr}
\and
Antonino Salibra
\institute{DAIS, Universit\`{a} Ca' Foscari Venezia, Italia
\thanks{Work partially supported by the Fondation de Math\'{e}matique de Paris.}}
\email{salibra@dsi.unive.it}
}

\def\titlerunning{Minimal lambda-theories by ultraproducts}
\def\authorrunning{A. Bucciarelli, A. Carraro \& A. Salibra}
\begin{document}
\maketitle

\begin{abstract}
A longstanding open problem in lambda calculus is whether there exist continuous models of the untyped lambda calculus whose theory is exactly
 the least lambda-theory $\lambda\beta$ or the least sensible lambda-theory $\cH$ (generated by equating all the unsolvable terms). A related question
 is whether, given a class of lambda models, there is a minimal lambda-theory represented by it. In this paper, we give a general tool to answer
 positively to this question and we apply it to a wide class of webbed models: the i-models. The method then applies also to graph models, Krivine models,
 coherent models and filter models. In particular, we build an i-model whose theory is the set of equations satisfied in all i-models.
\end{abstract}


\section{Introduction}

Lambda-theories are congruences on the set of $\gl$-terms which contain $\gb$-conversion,
providing (sound) notions of program equivalence. Models of the $\lambda$-calculus are one of the 
main tools used to study the lattice of $\lambda$-theories. 
After the first model,
 found by Scott in 1969 in the category of complete lattices and Scott continuous functions, a large number of mathematical models for $\gl$-calculus, arising from
 syntax-free constructions, have been introduced in various Cartesian closed categories (ccc, for short) of domains and were classified into semantics according to
 the nature of their representable functions, see e.g. \cite{Barendregt84,Berline00,Plotkin93}.
Scott continuous semantics \cite{Sco72} is the class of reflexive
 cpo-models, that are reflexive objects in the category $\CPO$, whose objects are complete partial orders and morphisms are Scott continuous functions. The stable
 semantics (Berry \cite{Ber78}) and the strongly stable semantics (Bucciarelli\textendash Ehrhard \cite{BucEhr91}) are refinements of the continuous
 semantics, introduced to approximate the notion of ``sequential'' Scott continuous function.

Some models of $\lambda$-calculus, called webbed models, are
built from lower level structures called ``webs'' (see Berline
\cite{Berline00} for an extensive survey). The simplest class of
webbed models is the class of graph models, which was isolated in the
seventies by Plotkin, Scott and Engeler
\cite{Engeler81,Plotkin93,Scott76} within the continuous
semantics. The class of graph models contains the simplest models of
$\gl$-calculus, is itself the easiest describable class, and
represents nevertheless a continuum of (non-extensional)
lambda-theories. Another example of a class of webbed models, and the
most established one, is the class of filter models. It was isolated
at the beginning of the eighties by Barendregt, Coppo and Dezani
\cite{BarendregtCD83}, after the introduction of the intersection type
discipline by Coppo and Dezani \cite{CoppoD80}. Not all filter models
live in Scott continuous semantics: for example some of them lack the
property of representing all continuous functions, and others were
introduced for the stable semantics (see Paolini et al. \cite{PPR09}, Bastonero
 et al. \cite{PravatoBR98}).


 In general, given a class $\cC$ of models, a natural completeness
 problem arises for it: whether the class is \emph{complete}, i.e.,
 for any lambda-theory $T$ there exists a member of $\cC$ whose
 equational theory is $T$. A related question, raised in
 \cite{Berline00} is the following: given a class $\cC$ of models of
 the $\gl$-calculus, is there a minimal lambda-theory represented by
 $\cC$? If this is the case, we say that $\cC$ enjoys the {\em
   minimality property}.  In \cite{DiGianantonioHP95} it was shown
 that the above question admits a positive answer for Scott's
 continuous semantics, at least if we restrict to extensional reflexive
 CPO-models.  Another result, in the same spirit, is the construction
 of a model whose theory is $\lambda\beta\eta$, {\em a fortiori}
 minimal, in the $\omega_1$-semantics (which is different from Scott semantics).  However, the proofs of
 \cite{DiGianantonioHP95} use logical relations, and since logical
 relations do not allow to distinguish terms with the same applicative
 behavior, the proofs do not carry over to non-extensional models.
 Similarly, in \cite{BucciarelliS08}, it is shown that the class of
 graph models enjoys the minimality property.

 In this paper, we propose a method to prove that a given class of
 models enjoys the minimality property, based on two main ingredients:
 the {\em finite intersection property} (fip) and the {\em
   ultraproduct property} (upp). The fip is
 satisfied by a class $\cC$ of models if for all models
 ${\bM}_1,{\bM}_2$ in $\cC$ there exists a model $\bM$ in $\cC$ whose equational 
 theory is included in $Th({\bM}_1)\cap Th({\bM}_2)$. The upp is satisfied in $\cC$ if for every non-empty family
 $\{\bM_i\}_{i\in I}$ of members of $\cC$ and for every proper
 ultrafilter $U$ of sets on $\cP(I)$ the ultraproduct $(\prod_{i\in
   I}\bM_i)/U$ can be embedded into a member of $\cC$. We show in
 Theorem \ref{thm:main-general} that if these conditions are
 satisfied, then $\cC$ has the minimality property.  An important
 technical device used in the proof of Theorem \ref{thm:main-general}
 is Lo\'{s} Theorem: the ultraproduct of a family of models
 satisfies an (in)equation between $\lambda$-terms if and only if the
 set of indexes of the component models satisfying it belongs to the
ultrafilter. Hence, proving the minimality property boils down to
exhibiting an appropriate ultrafilter.

As an application of this  general method, 
we prove that the class of i-models introduced in \cite{CarSal09}  enjoys the
minimality property. First of all, for every pair of i-models
$\bA,\bB$ we construct an i-model $\bC$ such that $Th(\bC)
\subseteq Th(\bA) \cap Th(\bB)$. This result is obtained via a completion process applied to the categorical product of
$\bA$ and $\bB$, adapted from \cite{CarSal09}.
 In order to prove that the class of i-models enjoys the upp, we exploit the fact that i-models are webbed models.
Given an ultraproduct $P$ of i-models, we construct the ultraproduct $P'$ of the corresponding webs. It turns out that $P'$ is a well defined
web. Then we show that there exists an embedding from $P$ to the i-model associated with $P'$.
We also show how our proof can be applied to smaller classes of webbed models,
like graph models, Krivine models, coherent models, and filter models.


Although we know that there exists a minimal i-model, its equational theory has not yet been characterized. Then the results of this paper do not 
 give a solution to the 
 longstanding open problem which asks whether there exist continuous models of the untyped lambda calculus whose theory is exactly the least
 $\lambda$-theory $\lambda\beta$.

The paper is organized as follows. In Section \ref{sec:preliminaries}
we provide the preliminary notions and results needed in the 
rest of the paper, in Section \ref{minimal-models} we present
the general method for showing that a given class of models 
of the $\lambda$-calculus has the minimality property,
and in Section \ref{applications} we apply this method to the
class of i-models.


\section{Preliminaries}\label{sec:preliminaries}
\subsection{Lambda-theories and models of lambda-calculus}\label{lambdacalculus}

With regard to the lambda-calculus we follow the notation and terminology of \cite{Barendregt84}. By $\gL$ and $\gLo$, respectively, we indicate the set of $\gl$-terms and of
 closed $\gl$-terms. We denote $\ga\gb$-conversion by $\gl\gb$. A {\em $\gl$-theory} is a congruence on $\gL$ (with respect to the operators of abstraction and
 application) which contains $\gl\gb$. A $\lambda$-theory is {\em consistent} if it does not equate all $\gl$-terms, {\em inconsistent} otherwise. The set of lambda-theories
 constitutes a complete lattice w.r.t. inclusion, whose top is the inconsistent lambda-theory and whose bottom is the theory $\lambda\beta$. The lambda-theory
 generated by a set $X$ of identities is the intersection of all lambda-theories containing $X$.

It took some time, after Scott gave his model construction, for consensus to arise on the general notion of a model of the $\lambda$-calculus. There are mainly two
 descriptions that one can give: the category-theoretical and the algebraic one. Besides the different languages in which they are formulated, the two approaches are
 intimately connected (see Koymans \cite{Koymans82}). The categorical notion of model, that of reflexive object in a Cartesian closed category (ccc), is well-suited for constructing concrete models, while the algebraic one is
 rather used to understand global properties of models (constructions of new models out of existing ones, closure properties, etc.) and to obtain results about the
 structure of the lattice of $\lambda$-theories. The main algebraic description of models of lambda-calculus is the class of \emph{$\gl$-models}, which are
 axiomatized over combinatory algebras by a finite set of first-order sentences (see Meyer \cite{Mey82}, Scott \cite{Sco80}, Barendregt \cite{Barendregt84}).
 In the following we denote by $\ssk$ and $\sss$ the so-called \emph{basic combinators}.

\subsection{Ultraproducts}\label{models-lam-calc}

Ultraproducts result from a suitable combination of the direct product and quotient constructions. They were introduced in the 1950's by Lo\'{s}.

Let $I$ be a non-empty set and let $\{\bA_i\}_{i \in I}$ be a family of combinatory algebras. Let $U$ be a proper ultrafilter of the boolean algebra $\cP(I)$.
 The relation $\sim_U$, given by $a \sim_U b \iff \{i \in I \st a(i) = b(i)\} \in U$, is a congruence on the combinatory algebra
 $\prod_{i\in I}\bA_i$. The \emph{ultraproduct} of the family $\{\bA_i\}_{i \in I}$, noted $(\prod_{i\in I}\bA_i)/U$, is defined as the quotient of the product
 $\prod_{i\in I}\bA_i$ by the congruence $\sim_U$. 
 If $a \in \prod_{i\in I}\bA_i$, then we denote by $a/U$ the equivalence class of $a$ with respect to the congruence $\sim_U$.
 If all members of $\{\bA_i\}_{i \in I}$ are $\lambda$-models, by a celebrated theorem of Lo\'{s} we have that
 $(\prod_{i\in I}\bA_i)/U$ is a $\lambda$-model too, because $\lambda$-models are axiomatized by first-order sentences.
 The basic combinators of the $\lambda$-model $(\prod_{i\in I}\bA_i)/U$ are $\ssk/U$ and $\sss/U$,
 and application is given by $x/U\cdot y/U = (x\cdot y)/U$, where the application $x\cdot y$ is defined pointwise.

We now recall the famous Lo\'{s} theorem that we will use throughout this paper.

\begin{theorem}[Lo\'{s}]\label{thm:Los}
Let $\cL$ be a first-order language and $\{\bA_i\}_{i\in I}$ be a family of $\cL$-structures indexed by a non-empty set $I$ an let $U$ be a proper ultrafilter
 of $\cP(I)$. Then for every $\cL$-formula $\varphi(x_1,\ldots,x_n)$ and for every tuple $(a_1,\ldots,a_n) \in \prod_{i \in I}\bA_i$ we have that
$$ (\prod_{i \in I}\bA_i)/U \models \varphi(a_1/U,\ldots,a_n/U) \iff \{i \in I \st \bA_i \models \varphi(a_1(i),\ldots,a_n(i))\} \in U. $$
\end{theorem}

\subsection{Information systems}\label{sec-information}

Information systems were introduced by Dana Scott in \cite{Sco82} to give a handy representation of Scott domains. An \emph{information system} is a tuple
 $\cA= (A,\Con_A,\vdash_A,\nu_A)$, where $A$ is a set and $\nu_A \in A$, $\Con_A \subseteq \cP_\rf(A)$ is a downward closed
 family containing all singleton subsets of $A$, and $\vdash_A\ \subseteq \Con_A \times A$ satisfies the four axioms listed below:
\begin{itemize}
\item[(I1)] if $a \in \Con_A$ and $a \vdash_A b$, then $a \cup b \in \Con_A$ \qquad (where $a \vdash_A b \stackrel{\mathrm{def}}{=} \forall \beta \in b.\ a \vdash_A \beta$)
\item[(I2)] if $\ga \in a$, then $a \vdash_A \ga$
\item[(I3)] if $a \vdash_A b$ and $b \vdash_A \gamma$, then $a \vdash_A \gamma$
\item[(I4)] $\emptyset \vdash_A \nu_A$
\end{itemize}

We adopt the following notational conventions: letters $\ga,\gb,\gc,\ldots$ are used for elements of $A$ (also called \emph{tokens}); letters $a,b,c,\ldots$ are used
 for elements of $\Con_A$, usually called \emph{consistent sets}; letters $x,y,z,\ldots$ are used for arbitrary elements of $\cP(A)$. We usually drop the
 subscripts from $\Con_A$ and $\vdash_A$ when there is no danger of confusion.

A subset $x \subseteq A$ is \emph{finitely consistent} if each of its finite subsets belongs to $\Con_A$.
 We denote by $\cP_{\mathrm{c}}(A)$ the set of all finitely consistent subsets of $A$. We define an operator $\cldn{}_A : \cP_{\mathrm{c}}(A) \to \cP_{\mathrm{c}}(A)$ by setting
 $\cldn{x}_{A} = \{\ga \in A : \exists a\subseteq_\rf x.\ a \vdash \ga \}$. We may drop the subscript when the underlying information system is clear from the context.
 Note that $\cldn{}$ is a monotone map satisfying the following conditions: $x \subseteq \cldn{x}$;\ $\cldn{\cldn{x}} = \cldn{x}$ and
 $\cldn{x} = \cup_{a \subseteq_\rf x\ } \cldn{a}$.
 We call \emph{point} any subset of $A$ which is in the image of $\cldn{}$. It is well-known that the set of points, partially ordered by
 inclusion, constitutes a Scott domain and any Scott domain is isomorphic to the set of points of some information system.
 

An \emph{approximable relation} between two information systems $\cA,\cB$ is a relation $R \subseteq \Con_A \times B$ satisfying the following properties:
\begin{itemize}
\item[(AR1)] if $a \in \Con_A$ and $a \ R \ b$, then $b \in \Con_B$ \qquad (where $a \ R \ b \stackrel{\mathrm{def}}{=} \forall \gb \in b.\ a \ R \ \gb$)
\item[(AR2)] if $a' \vdash_A a$, $a \ R \ b$, and $b \vdash_B \gb'$, then $a' \ R \ \gb'$.
\end{itemize}

$\Inf$ is the category which has information systems as objects and approximable relations as arrows. The composition of two morphisms
 $R \in \Inf(\cA,\cB)$ and $S \in \Inf(\cB,\cC)$ is (using the meta-notation) their usual relational composition:
 $S \circ R = \{(a,\gamma) \in \Con_A \times C : \exists b \in \Con_B.\ (a,b) \in R \text{ and } (b,\gamma) \in S \}$.
 The identity morphism of an information system $\cA$ is $\vdash_A$.

The Cartesian closed structure of $\Inf$ is described in \cite{Sco82}, and we recall it here for the sake of self-containment.

In what follows we use the projection functions $\fst$ and $\snd$ of a set-theoretic Cartesian product over the first and second component, respectively.
 The same notation is extended to finite subsets of the Cartesian product. For example, $\fst(a) = \{\fst(\ga) \st \ga \in a\}$.

\begin{definition}\label{def:cartesian-prod-inf-sys}
The Cartesian product of $\cA$ and $\cB$ is given by $\cA \binampersand \cB = (A \uplus B,\Con,\vdash,\nu)$ where
\[
\begin{array}{l}
A \uplus B = (\{\nu_A\} \times B) \cup (A \times \{\nu_B\}) \qquad \nu = (\nu_A,\nu_B) \qquad\qquad\qquad\qquad\qquad\qquad\qquad\qquad\qquad\qquad\qquad\qquad \\
a \in \Con \quad \text{ iff } \quad \pow{\fst}(a) \in \Con_A \text{ and } \pow{\snd}(a) \in \Con_B \\
a \vdash \ga \quad\ \ \ \text{ iff } \quad \pow{\fst}(a) \vdash_A \fst(\ga) \text{ and } \pow{\snd}(a) \vdash_B \snd(\ga)
\end{array}
\]
\end{definition}

The terminal object is the information system $\top$ whose underlying set contains only one token.

\begin{definition}\label{def:arrow-inf-sys}
The exponentiation of $\cB$ to $\cA$ is given by $\cA \Rightarrow \cB = (A \Rightarrow B,\Con,\vdash,\nu)$ where
\[
\begin{array}{l}
A \Rightarrow B = \Con_A \times B \qquad \nu = (\emptyset,\nu_B) \qquad\qquad\qquad\qquad\qquad\qquad\qquad\qquad\qquad\qquad\qquad\qquad\qquad \\
\{(a_1,\beta_1),\ldots,(a_k,\beta_k)\} \in \Con \quad\ \text{ iff } \quad \forall I\subseteq [1,k].\ (\cup_{i \in I} a_i \in \Con_A \Rightarrow \{\beta_i : i \in I\} \in \Con_B) \\
\{(a_1,\beta_1),\ldots,(a_k,\beta_k)\} \vdash (c,\gamma) \quad \text{ iff } \quad \{\beta_i : c \vdash_A a_i,\ i \in [1,k] \} \vdash_B \gamma
\end{array}
\]
\end{definition}

The category $\SD$ of Scott domains and Scott continuous functions is equivalent to the category $\Inf$ of information systems, via a pair of mutually
 inverse Cartesian closed functors $(\cdot)^+: \Inf \to \SD$ and $(\cdot)^-: \SD \to \Inf$.

In particular for an information system $\cA$, we have that $\cA^+$, the set of points of an information system, ordered by inclusion,
 is a Scott domain. Moreover, the domains $[\cA^+\to \cB^+]$ and $\cA^+ \times \cB^+$ are isomorphic (in the category $\SD$) to
 the domains $(\cA \Rightarrow \cB)^+$ and $(\cA \binampersand \cB)^+$, respectively.

\subsection{Webbed models of lambda-calculus}\label{webbed-information}



 
Let $\cA,\cB$ be information systems and let $f: A\to B$ be a function. We define two Scott continuous functions $f^\bullet: \cA^+ \to \cB^+$ and $f_\bullet: \cB^+ \to \cA^+$ as follows:
$$
f^\bullet(x) = \cldn{\{ f(\ga) \st \ga \in x \}}_B\ ; \qquad f_\bullet(y) = \cldn{\{\ga \st f(\ga) \in y \}}_A
$$
for every point $x$ of $\cA$ and every point $y$ of $\cB$. In \cite{CarSal09} simple conditions are given under which $f$ can generate a retraction pair
 $(f_\bullet,f^\bullet)$ from $\cA^+$ to $\cB^+$ in the category $\SD$, i.e., $f_\bullet \comp f^\bullet = \id_{\cA^+}$.
 

\begin{definition}[\cite{CarSal09}]\label{def:morph-ISs}
Let $\cA,\cB$ be information systems. A \emph{morphism} from $\cA$ to $\cB$ is a function $f: A \to B$ satisfying the following property:
\[\begin{array}{ll}
\text{(Mo)} & a \in \Con_A \text{ iff } \pow{f}(a) \in \Con_B \qquad\qquad\qquad\qquad\qquad\qquad\qquad\qquad\qquad\qquad\qquad\qquad
\end{array}\]
\end{definition}

\begin{definition}[\cite{CarSal09}]\label{def:b-f-morph-ISs}
A morphism $f: \cA \to \cB$ is a \emph{b-morphism} (resp.\ \emph{f-morphism}) if it satisfies the following property (bMo) (resp.\ (fMo))
\[\begin{array}{ll}
\text{(bMo)} & \text{if } \pow{f}(a) \vdash_B f(\ga) \text{, then } a \vdash_A \ga \qquad\qquad\qquad\qquad\qquad\qquad\qquad\qquad\qquad\qquad\qquad \\
\text{(fMo)} & \text{if } a \vdash_A \ga \text{, then } \pow{f}(a) \vdash_B f(\ga)
\end{array}\]
\end{definition}

The ``b" (resp.\ ``f") in the name of the axiom stands for backward (resp.\ forward). We leave to the reader the easy relativization of the various notions of
 morphism given in Definition \ref{def:b-f-morph-ISs} to the case in which $f$ is a partial map.

\begin{proposition}\label{homoiso}
 Let $f: \cA \to \cB$ be a b-morphism. Then $(f_\bullet,f^\bullet)$ is a retraction pair from $\cA^+$ into $\cB^+$. 
\end{proposition}

\begin{proof}
From (bMo) it follows $f_\bullet\circ f^\bullet = \id_{\cA^+}$.
\end{proof}

\begin{definition}
An \emph{i-web} is a pair $\bA = (\cA,\phi)$ where $\cA$ is an information system and $\phi:(\expsp{\cA}{\cA})\to\cA$ is a b-morphism.
\end{definition}

The set of tokens of $\cA$ is called the \emph{web} of $\bA$.

\begin{proposition}\label{cor-ref-obj}
Let $\bA = (\cA,\phi)$ be an i-web. Then $\cA^+$ is a reflexive object in the category $\SD$.
\end{proposition}

\begin{proof}
As anticipated, there is a continuous isomorphism $\theta:(\cA\Rightarrow\cA)^+\to [\cA^+\to\cA^+]$ and by Proposition \ref{homoiso} the domain $(\cA\Rightarrow\cA)^+$
 can be embedded into $\cA^+$ via the retraction pair $(\phi_\bullet,\phi^\bullet)$. Therefore $(\theta\circ\phi_\bullet,\theta^{-1}\circ\phi^\bullet)$
 is the desired retraction pair in the category $\SD$.
\end{proof}

We set $\bA^+ = (\cA^+,\theta\circ\phi_\bullet,\theta^{-1}\circ\phi^\bullet)$ and call $\bA^+$ an \emph{i-model}. Of course, since $\bA^+$ is a reflexive object
 in $\SD$, then $\bA^+$ is also a $\lambda$-model and closed $\lambda$-terms are interpreted as elements of $\cA^+$ (i.e.\ as points of $\cA$) 
as follows:
\[
\begin{array}{lcl}
\Int{x}_{\rho}^{\bA^+} & = & \rho(x), \text{ where $\rho$ is any map from $\Var$ into $\cA^+$} \\
\Int{\lambda y.M}_{\rho}^{\bA^+} & = & \cldn{\{\phi(a,\alpha) \st \alpha \in \Int{M}_{\rho[y:= \cldn{a}]}^{\bA^+}\}}_{A} \\
\Int{MN}_{\rho}^{\bA^+} & = & \{\beta \in A \st \exists a \subseteq_\rf \Int{N}_{\rho}^{\bA^+}.\ (a,\beta) \in \cldn{\{(a',\beta') \st \phi(a',\beta') \in \Int{M}_{\rho}^{\bA^+}\}}_{A\Rightarrow A} \}
\end{array}
\]
The $\lambda$-model structure associated to the i-model $\bA^+$ is the following. The basic combinators are
 $\ssk^{\bA^+} = \Int{\lambda xy.x}^{\bA^+}$ and $\sss^{\bA^+} = \Int{\lambda xyz.xz(yz)}^{\bA^+}$, and the application operation is given by
$$ u \cdot z = \{\beta \in A \st \exists a \subseteq_\rf z.\ (a,\beta) \in \cldn{\{(a',\beta') \st \phi(a',\beta')\in u\}}_{A\Rightarrow A} \} $$
for all points $u,z$.

\subsubsection{Well-known instances of i-webs}\label{particular-cases}

An \emph{extended abstract type structure} ($\EATS$, for short, \cite[Def.~1.1]{CopDezHonLon84}) is an algebra $(A,\wedge,\to,\go)$, where
 ``$\wedge$" and ``$\to$" are binary operations and ``$\go$" is a constant, such that $(A,\wedge,\go)$ is a meet-semilattice with
  top element $\go$. In the following $\leq$ denotes the partial order associated with the meet-semilattice structure.
Recall from \cite[Def.~2.12,Thm.~2.13]{CopDezHonLon84} that the filter models living in Scott semantics are obtained by taking the set of filters of $\EATS$s satisfying the
 following condition:
\begin{itemize}
\item[($\ast$)] If $\bigwedge_{i=1}^{n} (\ga_i \to \gb_i) \leq \gc \to \delta$, then $(\bigwedge_{i \in \{ i:  \gc\leq\ga_i\}} \gb_i) \leq \delta$.
\end{itemize} 

Given an $\EATS$ $(A,\wedge,\to,\go)$, the structure $\cA = (A,\cP_\rf(A),\vdash,\go)$, where $a \vdash \ga$ iff $(\bigwedge a) \leq \ga$, is an
 information system.
 
If the $\EATS$ satisfies condition ($\ast$), then the function $\phi: \cP_\rf(A)\times A \to A$ given by $\phi(a,\ga) = (\bigwedge a) \to \ga$
 is a b-morphism, and hence an i-web $\bA=(\cA,\phi)$. The
 corresponding filter model is exactly the i-model $\bA^+$ (see \cite{CarSal12} for the details). 

In Larsen and Winskel \cite{Larsen91} the definition of information system is slightly different: there is no special token $\nu$. We remark that the corresponding
 class of i-models generated by the two definitions is the same. We adopt Scott's original definition just for technical reasons. With Larsen \& Winskel's definition
 we can capture some other known classes of models, as illustrated below.

A \emph{preordered set with coherence} (pc-set, for short) is a triple $(A,\leq,\Bumpeq)$, where $A$ is a non-empty set, $\leq$ is a preorder on $A$ and $\Bumpeq$
 is a coherence (i.e., a reflexive, symmetric relation on $A$) compatible with the preorder (see \cite[Def.~120]{Berline00}). A pc-set ``is" an information system
 $\cA = (A,\cP_\rf^\mathrm{coh}(A),\vdash)$, where $\cP_\rf^\mathrm{coh}(A)$ is the set of finite coherent subsets of $A$ and
 $a \vdash \ga$ iff $\exists \gb \in a.\ \gb \geq \ga$. A \emph{pc-web} (see \cite[Def.~153]{Berline00}) is determined by a pc-set together with a map
 $\phi: \cP_\rf^\mathrm{coh}(A) \times A \to A$ satisfying:
\begin{enumerate}
\item[(1)] $\phi(a, \ga) \Bumpeq \phi(b, \gb)$ iff ($a\cup b\in \cP_\rf^\mathrm{coh}(A) \Rightarrow\! \ga \! \Bumpeq \! \gb $)
\item[(2)] if $\phi(a, \ga) \leq \phi(b, \gb)$, then $\ga \leq \gb$ and ($\forall \gamma\in b\ \exists\gd\in a. \gamma\leq\gd)$.
\end{enumerate}
A pc-web is a particular instance of i-web and properties (1),(2) say exactly that $\phi$ is a b-morphism. Krivine webs \cite[Sec.~5.6.2]{Berline00} are pc-webs in
 which $\Bumpeq\ = A \times A$ (so that $\cP_\rf^\mathrm{coh}(A) = \cP_\rf(A)$). \emph{Total pairs} \cite[Sec.~5.5]{Berline00} are Krivine webs in
 which $\leq$ is the equality: in fact in this the requirement of $\phi$ to be a b-morphism boils down to injectivity.
 Therefore a total pair is simply defined as a set $A$ together with an injection $i_A: \cP_\rf(A) \times A \to A$; the underlying
 information system is $\cA = (A,\cP_\rf(A),\ni)$. The \emph{graph model} associated to the total pair is then the i-model $\bA^+$, obtained
 by taking the powerset of $A$ (see \cite[Def.~120]{Berline00}). There is usually some ambiguity in the terminology since by ``graph model''
 sometimes is meant the total pair (as in \cite{BucSal03}, for example) underlying the model itself.


\section{Minimal models: general results}\label{minimal-models}

Given a class $\cC$ of $\lambda$-models, a natural question to be asked is whether there exists a member
 $\bA$ of $\cC$ such its equational theory, hereafter noted $\Th(\bA)$, is contained in the theories of all other members of $\cC$: one such model
 $\bA$ is called \emph{minimal} in $\cC$. This point was raised
 in print by C. Berline \cite{Berline00} who was mainly referring to the classes of webbed models of $\lambda$-calculus.
 If a positive answer is obtained, usually it is done by purely semantical methods and $\Th(\bA)$ does not need to 
 be characterised in the syntactical sense: this is the case of Di Giannantonio et al. \cite{DiGianantonioHP95}, in which the authors prove
 that the class all extensional reflexive CPOs has a minimal model. Of course if one is able to gather enough information about $\Th(\bA)$,
 then one may be in the position to answer the related completeness question for the class $\cC$: is $\lambda\beta$
 (or $\lambda\beta\eta$) a theory induced by a member of $\cC$? An example of result of this kind can be found again in
 \cite{DiGianantonioHP95}, where the authors construct a model with theory $\lambda\beta\eta$ in the $\omega_1$-semantics.

In this section we give general conditions for a class $\cC$ of $\lambda$-models under which we have the guarantee that 
 $\cC$ has a minimal model. In the forthcoming Section \ref{applications} we apply this general result to the class
 of i-models and some of its well-known classes of models.

\begin{definition}\label{def:fip}
A class $\cC$ of $\lambda$-models {\em has the finite intersection property} (fip, for short) if for every two 
 members $\bA$, $\bB$ of $\cC$, there exists a member $\bC$ of $\cC$ such that $\Th(\bC)\subseteq \Th(\bA)\cap \Th(\bB)$.
\end{definition}

For example the class of all $\lambda$-models has the fip, and in general every class closed under direct products
 has the fip. Every subclass which is axiomatized over the $\lambda$-models by first-order universal sentences has
 the fip, but of course these conditions do not hold in general for the classes of webbed models, e.g. for the
 i-models. We will see that they do hold for the filter models.

The fip is a property which is weaker than the closure under direct products. Of course a class which is closed under 
 arbitrary (non-empty) direct products has a minimal model. The next definition isolates a property that, together with
 the fip, can overcome the lack of direct products and guarantee the existence of minimal models.

\begin{definition}\label{def:up}
A class $\cC$ of $\lambda$-models {\em has the ultraproduct property} (upp, for short) if for every non-empty family
 $\{\bA_i\}_{i\in I}$ of members of $\cC$ and for every proper ultrafilter $U$ of sets on $\cP(I)$  
 the ultraproduct $(\prod_{i\in I}\bA_i)/U$ can be embedded into a member of $\cC$.
\end{definition}

For example the class of all $\lambda$-models has the upp, and in general every class closed under ultraproducts
 has the upp. Every subclass which is axiomatized over the $\lambda$-models by first-order sentences has
 the upp, but of course these conditions do not hold in general for the known classes of webbed models, e.g. for the
 i-models.

\begin{theorem}\label{thm:main-general}
Let $\cC$ be a class of $\lambda$-models having both the fip and the upp. Then $\cC$ has a minimal model.
\end{theorem}

\begin{proof}
Let $I$ be the set of all equations $e$ betweeen closed combinatory terms for which there exists a model $\bA$ in $\cC$
 such that $\bA \not\models e$.
 For every $e \in I$, consider the set $K_e =\{J \subseteq_\rf I \st e \in J\}$.
 Since $K_e \cap K_{e'} = \{ J \subseteq_\rf I \st e,e'\in J\} \neq \emptyset$ for all $e,e'\in I$, then there exists a
 non-principal ultrafilter $U$ on $\cP_\rf(\cP_\rf(I))$ containing the family $(K_e \st e \in I)$.  
 By the finite intersection property of the class $\cC$, for every $J \subseteq_\rf I$ there exists a model
 $\bA_J$ in $\cC$ such that $e\not\in \Th(\bA_J)$ for every $e\in J$. 
 Let $\{\bA_J\}_{J \subseteq_\rf I}$ be the family composed by these models and consider the ultraproduct
 $\bP_U = (\prod_{J \subseteq_\rf I}\bA_J)/U$. Let $e \in I$ be a closed equation and let
 $X_e = \{J \subseteq_\rf I \st \bA_J \not\models e\}$. Then $X_e \supseteq K_e \in U$, so that
 $X_e$ belongs to the ultrafilter $U$. Since $e$ is a closed first-order formula, by Lo\'{s} Theorem \ref{thm:Los}
 $\bP_U \not\models e$. Since $e$ was an arbitrary equation in $I$, we have that $\bP_U \not\models e$ for every $e\in I$,
 so that $\Th(\bP_U) \subseteq \bigcap_{\bA \in \cC} \Th(\bA)$.
 Finally, since the class $\cC$ has the ultraproduct property, then there exists a model $\bB$ in $\cC$ such that $\bP_U$
 embeds into $\bB$. Then $\Th(\bB) = \Th(\bP_U) \subseteq \bigcap_{\bA \in \cC} \Th(\bA) \subseteq \Th(\bB)$ and we get
 the desired conclusion.
\end{proof}

\begin{corollary}\label{cor:main-general}
Let $\cC$ be a class of $\lambda$-models which has the fip and is closed under ultraproducts.
 Then $\cC$ has a minimal model.
\end{corollary}

We conclude the section by giving some other general results that can be proved by just assuming the fip and the upp
 for a class $\cC$ of $\lambda$-models. In particular we prove a compactness theorem for lambda-theories whose
 equations hold in members of $\cC$. We also prove that, if there exists an easy $\lambda$-term in $\cC$, then there exists a continuum of 
 different equational $\cC$-theories. In other words, there are uncountably many different lambda-theories induced by models of the class $\cC$.

\begin{theorem}[Compactness]\label{compact}
Let $\cC$ be a class of $\lambda$-models having the upp, and let $E$ be a set of equations between closed $\lambda$-terms.
 If every finite subset of $E$ is satisfied by a member of $\cC$, then $E$ itself is satisfied by a member of $\cC$.
\end{theorem}

\begin{proof}
For every $e \subseteq_\rf E$, let $K_e = \{d\subseteq_\rf E \st e \subseteq d\}$ and let $\bA_e \in \cC$ be a model
 satisfying $e$. Let $U$ be a proper ultrafilter on $\cP_\rf(\cP_\rf(E))$ containing $K_e$ for every $e\subseteq_{\rf} E$.
 Then the ultraproduct $(\prod_{e \subseteq_\rf E}\bA_e)/U$ satisfies $E$. Finally by the upp there exists a model $\bB$ in $\cC$ 
 such that $(\prod_{e \subseteq_\rf E}\bA_e)/U$ embeds into $\bB$, and thus has the same lambda-theory. We conclude that $\bB$ satisfies $E$.
\end{proof}

Let $\cC$ be a class of $\lambda$-models. A closed $\lambda$-term $M$ is \emph{$\cC$-easy} if for every closed
 $\lambda$-term $N$ there exists a member $\bB$ of $\cC$ such that $\Int{M}^{\bB} = \Int{N}^{\bB}$.

\begin{theorem}\label{easy}
Let $\cC$ be a class of $\lambda$-models having the upp such that there exists a $\cC$-easy $\lambda$-term.
 Then there exist uncountably many $\cC$-theories. 
\end{theorem}

\begin{proof}
Let $M$ be a $\cC$-easy $\lambda$-term. For $n \geq 1$, we let $\gp_n\equiv\gl x_1\ldots x_n.x_n$.
 We prove that for every $n \geq 1$ the term $M\gp_n$ is $\cC$-easy. 

Let $X = (N_n)_{n\geq 1}$ be an arbitrary infinite sequence of closed $\beta\eta$-normal $\lambda$-terms and define
 $E(X) =\{M\gp_n = N_n \st n\geq 1\}$.
 Let $K =\{M\gp_{n_1}=N_{n_1},\ldots,M\gp_{n_k}=N_{n_k}\}$ be a finite subset of $E(X)$. Without loss of generality, we may
 assume that $n_1<\cdots<n_k$. Let $y$ be a fresh variable and define inductively
$$ Z_1 := y\underbrace{\bI\cdots\bI}_{n_1-1}N_{n_1}\ ; \qquad Z_{m+1} := Z_m\underbrace{\bI\cdots\bI}_{n_{m+1}-n_{m}-1}N_{n_m} $$ 
 Now set $Z=\gl y.Z_k$. Since $M$ is $\cC$-easy, then there is a member $\bA$ of $\cC$ such that $\bA \models M = Z$.
 Therefore $\bA \models M\gp_{n_i} = Z\gp_{n_i} = N_{n_i}$ for all $i = 1,\ldots,k$ so that
 $K \subseteq \Th(\bA)$. Since every finite subset of $E(X)$ is satisfied by a member of $\cC$, then by Theorem \ref{compact}
 $E(X)$ itself is satisfied by a member of $\cC$, i.e. there exists a member $\bA_X$ of $\cC$ such that 
 $E(X) \subseteq \Th(\bA_X)$. Moreover if $X$ and $Y$ are two different infinite sequences of closed $\beta\eta$-normal $\lambda$-terms,
 then $\Th(\bA_X) \neq \Th(\bA_Y)$. The result then follows from the fact that there are uncountably many infinite sequences of
 closed $\beta\eta$-normal $\lambda$-terms.
\end{proof}

\section{Applications}\label{applications}

In the present section we apply the general results developed in Section \ref{minimal-models}. In particular we prove that the class of i-models has
 both the finite intersection property and the ultraproduct property. Then we comment on how these general results also apply to other well-known
 classes of webbed models.

\subsection{Finite intersection property for i-models}\label{fip-imodels}

The goal of the first part of this section is to prove that for every pair $\bA_1,\bA_2$ of i-webs there exists an i-web $\bB$ such that
 $\Th(\bB^+) \subseteq \Th(\bA_1^+) \cap \Th(\bA_2^+)$. Such result would be trivial if the categorical product $\cA_1 \binampersand \cA_2$ could always be endowed with
 a suitable structure of i-web, but this is not the case. The best that we can do in general is to make $\cA_1 \binampersand \cA_2$ into a \emph{partial i-web}. A
 partial i-web in general is a pair $\bA = (\cA,\phi_A)$, where $\phi_A: \cA \Rightarrow \cA \rightharpoonup \cA$ is a partial b-morphism. In particular,
 $\bA_1\binampersand\bA_2$ is a partial i-web if we set
 if we can set 
$$
\phi(a,\ga) = 
\begin{cases}
(\nu_{A_1},\nu_{A_2}) & \text{ if } a \subseteq \{(\nu_{A_1},\nu_{A_2})\} \text{ and } \ga = (\nu_{A_1},\nu_{A_2}) \\
(\nu_{A_1},\phi_{A_2}(\pow{\snd}(a),\snd(\ga))) & \text{ if } a \cup \{\ga\} \subseteq_\rf \{\nu_{A_1}\} \times A_2 \\
(\phi_{A_1}(\pow{\fst}(a),\fst(\ga)),\nu_{A_2}) & \text{ if } a \cup \{\ga\} \subseteq_\rf A_1 \times \{\nu_{A_2}\}
\end{cases}
$$
A partial i-web does not give in general an i-model, but we can complete it to an i-web through a limit process that involves countably many extension steps.
 
We say that $\cB$ is an \emph{extension} of $\cS$, notation $\cS \preceq \cB$, if $S \subseteq B$, $\Con_S = \Con_B \cap \cP_\rf(S)$,
 $\vdash_S\ = \ \vdash_B \cap (\Con_S \times S)$. We say that $\bB$ is an extension of $\bS$, notation $\bS \preceq \bB$, if $\cS \preceq \cB$ and $\phi_S$ is the
 restriction of $\phi_B$ to $\Con_S \times S$. 

Let us call $\bB$ the result of the (yet undefined) completion process of $\bA_1\binampersand\bA_2$. Of course $\bB$ must be somehow related to the original i-webs $\bA_1$ and $\bA_2$.
 In particular, we want that for every closed $\gl$-term $M$ if $(\nu_{A_1},\gb) \in \Int{M}^{\bB^+}$ (resp.\ $(\ga,\nu_{A_2}) \in \Int{M}^{\bB^+}$), then $\gb \in \Int{M}^{\bA_2^+}$
 (resp.\ $\ga \in \Int{M}^{\bA_1^+}$) because this will guarantee that $\Th(\bB^+) \subseteq \Th(\bA_1^+) \cap \Th(\bA_2^+)$. We will achieve this property by means of the
 notion of \emph{f-morphism} of partial i-webs.

\textbf{Notation.} Let $f: A \rightharpoonup B$ be a partial function. We write $\dom(f)$ to indicate the domain of $f$
 and $\overline{\dom}(f)$ to indicate the complement of $\dom(f)$ in $A$. We define
 $\pow{f}: \cP_\rf(B) \to \cP_\rf(C)$ and $\pro{f}: (\cP_\rf(B) \times B) \to (\cP_\rf(C) \times C)$ as
 follows: $\pow{f}(b) = \{f(\gb) \mid \gb \in b,\ \gb \in \dom(f) \}$ and $\pro{f}(b,\gb) = (\pow{f}(b),f(\gb))$. Hence
 $\pow{\pro{f}}: \cP_\rf(\cP_\rf(B) \times B) \to \cP_\rf(\cP_\rf(C) \times C)$.

\begin{definition}[\cite{CarSal09}]\label{def:log-map}
Let $\bB,\bC$ be partial i-webs. An \emph{f-morphism from $\bB$ to $\bC$} is an f-morphism $\psi: \cB \to \cC$ satisfying the following additional property:
\[\begin{array}{ll}
\text{(iMo)} & \text{if } (a,\beta) \in \dom(\phi_B) \text{, then } (\pow{\psi}(a),\psi(\beta)) \in \dom(\phi_C) \text{ and } \psi(\phi_B(a,\beta)) = \phi_C(\pow{\psi}(a),\psi(\beta)) \\
\end{array}\]
\end{definition}

The following proposition explains that, in general, f-morphisms of i-webs ``commute'' well to the interpretation of $\lambda$-terms.  

\begin{proposition}[\cite{CarSal09}]\label{prop:f-mor-interp}
Let $\bB,\bC$ be i-webs, let $\psi:\bB \to \bC$ be an f-morphism of i-webs, and let $M$ be a closed $\lambda$-term.
 If $\ga \in \Int{M}^{\bB^+}$, then $\psi(\ga) \in \Int{M}^{\bC^+}$.
\end{proposition}

We remark that the two projection functions $\fst$ and $\snd$ are f-morphisms of partial i-webs from $\bA_1\binampersand\bA_2$ to $\bA_1$ and $\bA_2$, respectively.

Our goal now is to construct a series of triples $\{(\bS_n,\psi_n^1,\psi_n^2)\}_{n\geq 0}$ such that $\bS_{n} \preceq \bS_{n+1}$ and $\psi_{n}^i: \bS_{n} \to \bA_i$
 ($i=1,2$) is an f-morphism of partial i-webs such that $\psi_{n+1}^i$ extends $\psi_{n}^i$ ($i=1,2$). The idea is that the input parameter of the whole construction
 is the triple $(\bS_0,\psi_0^1,\psi_0^2)$ where $\bS_0:=\bA_1 \binampersand \bA_2$, $\psi_0^1 = \fst$, and $\psi_0^2 = \snd$. All subsequent triples are constructed
 via an algorithm that, given $(\bS_n,\psi_n^1,\psi_n^2)$ as input, returns $(\bS_{n+1},\psi_{n+1}^1,\psi_{n+1}^2)$. The union of all partial i-webs and all
 f-morphisms of partial i-webs finally gives an i-web $\bS_\go$ (called \emph{completion}) and two f-morphisms $\psi_\go^i$ ($i=1,2$) of i-webs that allow to show
 that $\Th(\bS_\go^+) \subseteq \Th(\bA_1^+) \cap \Th(\bA_2^+)$.

The $0$-th stage of the completion process, i.e., the triple $(\bS_0,\psi_0^1,\psi_0^2)$ has already been described. Now assuming we reached stage $n$, we show how to
 carry on with stage $n+1$.

\begin{definition}\label{step-of-construction}
\begin{itemize}
\item $S_{n+1} = S_n \cup \overline{\dom}(\phi_{S_n})$
\item $\Con_{S_{n+1}}$ is the smallest family of sets $x \subseteq_\rf S_n \cup \overline{\dom}(\phi_{S_n})$ such that either 
\begin{itemize}
\item[(1)] there exist $a \in \Con_n$ and $X \in \Con_{S_n \Rightarrow S_n}$ such that $X \subseteq \overline{\dom}(\phi_{S_n})$ and $x = a \cup X$ and 
 $\pow{\psi_{n}^i}(a) \cup \pow{\phi_{A_i}}(\pow{\pro{\psi_{n}^i}}(X)) \in \Con_{A_i}$ ($i=1,2$) or
\item[(2)] there exists $X \in \Con_{S_n\Rightarrow S_n}$ such that
 $x \subseteq_\rf (X \cap \overline{\dom}(\phi_{S_n})) \cup \cldn{(\pow{\phi_{S_n}}(X \cap \dom(\phi_{S_n})))}_{S_n}$
\end{itemize}
\item $a \vdash_{S_{n+1}} \alpha$ iff either $a \cap S_n \vdash_{S_n} \alpha$ or $\alpha \in a$
\item $\nu_{S_{n+1}} = \nu_{S_{n}}$
\item $\phi_{S_{n+1}}(a,\alpha) = 
\begin{cases}
\phi_{S_{n}}(a,\alpha) & \text{if } (a,\alpha) \in \dom(\phi_{S_{n}}) \\
(a,\alpha)             & \text{if } (a,\alpha) \in \overline{\dom}(\phi_{S_{n}}) \\
\text{undefined}       & \text{if } (a,\alpha) \in (S_{n+1}\Rightarrow S_{n+1})-(S_{n}\Rightarrow S_{n})
\end{cases}$
\item for $i=1,2$ we set 
$\psi_{n+1}^i(\alpha) = 
\begin{cases}
\psi_{n}^i(\alpha)                            & \text{if } \alpha \in S_n \\
\phi_{A_i}(\pow{\psi_n^i}(b),\psi_n^i(\beta)) & \text{if } \alpha = (b,\beta) \in S_{n+1} - S_n
\end{cases}$
\end{itemize}
\end{definition}

\begin{theorem}\label{step-of-completion}
We have that
\begin{itemize}
\item[(i)]   $\cS_{n+1} = (S_{n+1},\Con_{S_{n+1}},\vdash_{S_{n+1}},\nu_{S_{n+1}})$ is an information system such that $\cS_n \preceq \cS_{n+1}$,
\item[(ii)]  $\bS_{n+1} = (\cS_{n+1},\phi_{S_{n+1}})$ is a partial i-web such that $\bS_n \preceq \bS_{n+1}$,
\item[(iii)] $\psi_{n+1}^i: \bS_{n+1} \to \bA_i$ ($i=1,2$) is an f-morphism of partial i-webs.
\end{itemize}
\end{theorem}

\begin{proof}
\noindent(i) We show that $\cS_{n+1}$ is an information system, checking the properties (I1)-(I4) (see beginning of Section \ref{sec-information}).
\begin{itemize}
\item[(I1)] Suppose $a \in \Con_{S_{n+1}}$ and $a \vdash_{S_{n+1}} b$. 
If $a$ has been added to $\Con_{S_{n+1}}$ by clause (1), then exists $i \in \{1,2\}$, $a' \in \Con_{S_n}$ and $X \in \Con_{S_n \Rightarrow S_n}$ such that
 $X \subseteq \overline{\dom}(\phi_{S_n})$ and $a = a' \cup X$ and $\pow{\psi_{n}^i}(a') \cup \pow{\phi_{A_i}}(\pow{\pro{\psi_{n}^i}}(X)) \in \Con_{A_i}$.
 Since $a \vdash_{S_{n+1}} b$, then $b = b' \cup X$, for some $b' \in \Con_{S_n}$ such that $a' \vdash_{S_{n}} b'$. Now $\psi_n^i$ is a morphism, so that
 $\pow{\psi_{n}^i}(b') \cup \pow{\phi_{A_i}}(\pow{\pro{\psi_{n}^i}}(X)) \in \Con_{A_i}$. Therefore $b$ is added to $\Con_{S_{n+1}}$ by clause (1).

If $a$ has been added to $\Con_{S_{n+1}}$ by clause (2), then also $b$ is added to $\Con_{S_{n+1}}$ by the same clause.

\item[(I2)] If $\ga \in a$, then $a \vdash_{S_{n+1}} \ga$ by definition of $\vdash_{S_{n+1}}$.
\item[(I3)] Suppose $a \vdash_{S_{n+1}} \{\alpha_1,\ldots,\alpha_k\}$ and $\{\alpha_1,\ldots,\alpha_k\} \vdash_{S_{n+1}} \gamma$.
 If $\gamma \in \{\alpha_1,\ldots,\alpha_k\}$ then clearly $a \vdash_{S_{n+1}} \gamma$. Otherwise $\{\alpha_1,\ldots,\alpha_k\} \cap S_n \vdash_{S_{n}} \gamma$
 and since $a \cap S_n \vdash_{S_{n}} \{\alpha_1,\ldots,\alpha_k\} \cap S_n$ we can conclude using the property (I3) of $\cS_n$.
\item[(I4)] Immediate.
\end{itemize}
Finally it is immediate to see that $\cS_n \preceq \cS_{n+1}$.

\noindent(ii) Note that the fact that $\cS_n \preceq \cS_{n+1}$ automatically implies $\cS_n \Rightarrow \cS_n \preceq \cS_{n+1} \Rightarrow \cS_{n+1}$.
 Now we prove that $\phi_{S_{n+1}}: \cS_{n} \Rightarrow \cS_{n} \to \cS_{n+1}$ is a total b-morphism, so that it is automatically a partial b-morphism
 from $\cS_{n+1} \Rightarrow \cS_{n+1}$ to $\cS_{n+1}$. 
\begin{itemize}
\item[(Mo)] We must show that $X \in \Con_{S_n \Rightarrow S_n}$ iff $(X\cap\overline{\dom}(\phi_{S_n}))\cup(\pow{\phi_{S_n}}(X\cap\dom(\phi_{S_n}))) \in \Con_{S_{n+1}}$.
 If $X \in \Con_{S_n \Rightarrow S_n}$, then $(X\cap\overline{\dom}(\phi_{S_n}))\cup(\pow{\phi_{S_n}}(X\cap\dom(\phi_{S_n})))$ is in $\Con_{S_{n+1}}$ by clasuse (2).

Let $x = (X\cap\overline{\dom}(\phi_{S_n}))\cup(\pow{\phi_{S_n}}(X\cap\dom(\phi_{S_n}))) \in \Con_{S_{n+1}}$. If $x$ is added to $\Con_{S_{n+1}}$ by clause (1), then
 there exist $i \in \{1,2\}$, $a \in \Con_n$ and $Y \in \Con_{S_n \Rightarrow S_n}$ such that $Y \subseteq \overline{\dom}(\phi_{S_n})$ and $x = a \cup Y$ and 
 $\pow{\psi_{n}^i}(a) \cup \pow{\phi_{A_i}}(\pow{\pro{\psi_{n}^i}}(Y)) \in \Con_{A_i}$. Therefore $Y = (X\cap\overline{\dom}(\phi_{S_n}))$ and
 $a = (\pow{\phi_{S_n}}(X\cap\dom(\phi_{S_n})))$. Now we have
\[
\begin{array}{lcl}
\pow{\psi_{n}^i}(a) \cup \pow{\phi_{A_i}}(\pow{\pro{\psi_{n}^i}}(Y)) 
 & = & \pow{\psi_{n}^i}((\pow{\phi_{S_n}}(X\cap\dom(\phi_{S_n})))) \cup \pow{\phi_{A_i}}(\pow{\pro{\psi_{n}^i}}((X\cap\overline{\dom}(\phi_{S_n})))) \\
 & = & \pow{\phi_{A_i}}(\pow{\pro{\psi_{n}^i}}((X\cap\dom(\phi_{S_n})))) \cup \pow{\phi_{A_i}}(\pow{\pro{\psi_{n}^i}}((X\cap\overline{\dom}(\phi_{S_n})))) \\
 & = & \pow{\phi_{A_i}}(\pow{\pro{\psi_{n}^i}}(X))
\end{array}
\]
Since $\pow{\psi_{n}^i}(a) \cup \pow{\phi_{A_i}}(\pow{\pro{\psi_{n}^i}}(Y))$ is in $\Con_{A_i}$ by hypothesis, then so is $\pow{\phi_{A_i}}(\pow{\pro{\psi_{n}^i}}(X))$ and since both $\phi_{A_i}$ and $\psi_{n}^i$ are morphisms
 of information systems, then so is their composition $\phi_{A_i} \circ \psi_{n}^i$, meaning that $X \in \Con_{S_n \Rightarrow S_n}$.

If $x$ is added to $\Con_{S_{n+1}}$ by clause (2), then evidently $X \in \Con_{S_n \Rightarrow S_n}$.

\item[(bMo)] We must show that $\pow{\phi_{S_{n+1}}}(X) \vdash_{S_{n+1}} \phi_{S_{n+1}}(a,\ga)$ implies $X \vdash_{S_{n+1}\Rightarrow S_{n+1}} (a,\ga)$.
 There are two cases to be dealt with. If $(a,\ga) \in \dom(\phi_{S_{n}})$, then $\pow{\phi_{S_{n}}}(X) \cap S_n \vdash_{S_n} \phi_{S_n}(a,\ga)$ and we derive
 $\pow{\phi_{S_{n}}}(X \cap \dom(\phi_{S_{n}})) \vdash_{S_n} \phi_{S_{n}}(a,\ga)$ so that by (bMo) for $\phi_{S_n}$ we have that
 $X \cap \dom(\phi_{S_n}) \vdash_{S_n\Rightarrow S_n} (a,\ga)$ and hence $X\vdash_{S_n\Rightarrow S_n} (a,\ga)$.

If $(a,\ga) \not\in \dom(\phi_{S_n})$, then $(a,\ga) = \phi_{S_{n+1}}(a,\ga) \in \pow{\phi_{S_{n+1}}}(X)$, so that $(a,\ga) \in X$ and thus
 $X\vdash_{S_n\Rightarrow S_n} (a,\ga)$.
\end{itemize}

\noindent(iii) Now we prove that $\psi_{n+1}^i$ $(i=1,2)$ is an f-morphism of i-webs.
\begin{itemize}
\item[(Mo)] ($\Rightarrow$) Suppose $x \in \Con_{S_{n+1}}$. We consider the clauses (1) and (2) of the definition of $\Con_{S_{n+1}}$.

If $x$ is added by clause (1), i.e. $x = a \cup X$ for suitable $a$ and $X$, then
 $\pow{\psi_{n+1}^i}(x) = \pow{\psi_{n}^i}(a) \cup \pow{\phi_{A_i}}(\pow{\pro{\psi_n^i}}(X)) \in \Con_{A_i}$, by clause (1) itself.

If $x$ is added by clause (2), then $x \subseteq_\rf (X \cap \overline{\dom}(\phi_{S_n})) \cup \cldn{(\pow{\phi_{S_n}}(X \cap \dom(\phi_{S_n})))}_{S_n}$,
 for some $X \in \Con_{S_n\Rightarrow S_n}$. Now let $y = (X \cap \overline{\dom}(\phi_{S_n})) \cup \pow{\phi_{S_n}}(X \cap \dom(\phi_{S_n}))$. We first observe that
\[
\begin{array}{ll}
\pow{\psi_{n+1}^i}(y) & = \pow{\phi_{A_i}}(\pow{\pro{\psi_{n}^i}}(X \cap \overline{\dom}(\phi_{S_n}))) \cup \pow{\psi_n^i}(\pow{\phi_{S_n}}(X \cap \dom(\phi_{S_n}))) \\
& = \pow{\phi_{A_i}}(\pow{\pro{\psi_n^i}}(X \cap \overline{\dom}(\phi_{S_n}))) \cup \pow{\phi_{A_i}}(\pow{\pro{\psi_n^i}}(X \cap \dom(\phi_{S_n}))) \\
& = \pow{\phi_{A_i}}(\pow{\pro{\psi_n^i}}(X))
\end{array}
\]
This proves that $\pow{\psi_{n+1}^i}(y) \in \Con_{A_i}$. Now using property (fMo) $\psi_{n}^i$ we obtain that $\pow{\psi_{n+1}^i}(y) \vdash_{A_i} \pow{\psi_{n+1}^i}(x)$,
 and hence $\pow{\psi_{n+1}^i}(x) \in \Con_{A_i}$.

($\Leftarrow$) By the very definition of $\Con_{S_{n+1}}$, in particular by the clause (1).

\item[(fMo)] Suppose $a \vdash_{S_{n+1}} \ga$. If $\ga \in a$, then of course $\pow{\psi_{n+1}^i}(a) \vdash_{A_{i}} \psi_{n+1}^i(\ga)$.
 If $a \cap S_n \vdash_{S_{n}} \ga$, then
$$
\pow{\psi_{n+1}^i}(a) = \pow{\psi_{n+1}^i}(a - S_n) \cup \pow{\psi_{n}^i}(a \cap S_n)
                      \vdash_{A_{i}} \pow{\psi_{n}^i}(a \cap S_n)
                      \vdash_{A_{i}} \psi_{n}^i(\ga)
$$
\item[(iMo)] Let $(a,\alpha) \in S_n\Rightarrow S_n$. Then $\psi_{n+1}^i(\phi_{S_{n+1}}(a,\alpha)) = \phi_{A_i}(\pow{\psi_n^i}(a),\psi_n(\alpha)) = \phi_{A_i}(\psi_{n+1}^i(a),\psi_{n+1}^i(\alpha))$,
 by definition of $\psi_{n+1}^i$ and the fact that it extends $\psi_n^i$.
\end{itemize}
\end{proof}

The \emph{completion} of the triple $(\bA_1 \binampersand \bA_2,\pi_1,\pi_2)$ is the triple $(\bS_\go,\psi_\go^1,\psi_\go^2)$, where 
 $\cS_\go = (S_\go,\Con_{S_\go},\vdash_{S_\go},\nu_{S_\go})$ and $\bS_\go = (\cS_\go,\phi_{S_\go})$ are given by the following data:
\[
\begin{array}{lll}
S_\go := \bigcup_{m < \omega} S_m & \Con_{S_\go}:= \bigcup_{m < \omega} \Con_{S_m} & \vdash_{S_\go} := \bigcup_{m < \omega} \vdash_{S_m} \\
\nu_{S_\go} := \nu_{A_1 \binampersand A_2} & \phi_{S_\go} := \bigcup_{m < \omega} \phi_{S_m} & \psi_\go^i := \bigcup_{m < \omega} \psi_m^i\ (i=1,2)
\end{array}
\]

\begin{lemma}\label{lem:completio-final}
$\bS_\go$ is an i-web and $\psi_\go^i: \bS_\go \to \bA_i$ ($i=1,2$) is an f-morphism of i-webs.
\end{lemma}

\begin{proof}
Indeed $\cS_\go$ is an information system as a consequence of Theorem \ref{step-of-completion}(i). Moreover the map $\phi_{S_\go}$ is total and it is easy to prove that it is a b-morphism from $\cS_\go \Rightarrow \cS_\go$
 using the fact that for every $n$ the map $\phi_{S_{n+1}}$ is a partial b-morphism (Theorem \ref{step-of-completion}(ii)). Similarly one can prove that $\psi_\go^i$ is an f-morphism of i-webs from
 $\bS_\go$ to $\bA_i$ ($i=1,2$) simply using the fact that for every $n$ the map $\psi_{n}^i$ is an f-morphism from the partial i-web $\bS_n$ to the i-webs
 $\bA_i$ ($i=1,2$) (Theorem \ref{step-of-completion}(iii)).
\end{proof}

\begin{theorem}\label{theory-contained}
$\Th(\bS_\go^+) \subseteq \Th(\bA_1^+) \cap \Th(\bA_2^+)$.
\end{theorem}

\begin{proof}
Suppose $M = N \not\in \Th(\bA_1^+) \cap \Th(\bA_2^+)$. Suppose, w.l.o.g., that $M = N \not\in \Th(\bA_1^+)$. Then there exists $\ga \in A_1$ such that
 $\ga \in \Int{M}^{\bA_1^+} - \Int{N}^{\bA_2^+}$. It is not difficult to check that $\ga \in \Int{M}^{\bA_1^+}$ implies $(\ga,\nu_{A_2}) \in \Int{M}^{\bS_\go^+}$, since
 $\bS_\go$ extends $\bA_1 \binampersand \bA_2$. Now suppose, by way of contradiction, that $(\ga,\nu_{A_2}) \in \Int{N}^{\bS_\go^+}$. Since
 $\psi_\go^1(\ga,\nu_{A_2}) = \ga$, by Proposition \ref{prop:f-mor-interp} we have that $\ga \in \Int{N}^{\bA_1^+}$, which is a contradiction.
 This proves that $(\ga,\nu_{A_2}) \in \Int{M}^{\bS_\go^+}-\Int{N}^{\bS_\go^+}$, so that $M = N \not\in \Th(\bS_\go^+)$.
\end{proof}

In Section \ref{particular-cases} we indicate how some of the most known classes of webbed models are recovered as particular instances of i-models (more details for Filter Models
 are in \cite{CarSal12}). Along these lines the notion of partial i-web generalizes those of \emph{partial pair} \cite{Berline00} (related to graph models)
 as well as the notions of partial webs of the other types.

The idea of partial pair and of a completion for obtaining a graph model generalizes the construction of the Engeler model and the of the Plotkin\textendash Scott $\cP_\go$
 model. It was initiated by Longo in \cite{Longo83} and further developed and applied by Kerth \cite{Kerth98b}.
 Definition \ref{step-of-construction} is the core of a completion of i-webs that further generalizes Longo and Kerth's work.
 As such, it can be adapted case by case so that the entire completion adapts to the various instances of i-webs in the sense that if we start
 with partial pair, at the end we obtain a total pair, if we start with a partial pcs-web, we end up in a total pcs-web etc.

Of course Theorem \ref{theory-contained} proves the finite intersection property for the class of i-models, but in view of the above discussion
 it can also give proofs of the finite intersection property for the subclasses of models mentioned in section \ref{particular-cases}.
 
For the particular case of graph models the fip was proved by Bucciarelli\&Salibra \cite{BucciarelliS08,BucSal03}, via a construction that they call
 \emph{weak product} which has the same spirit of our completion method. For the other classes of models the fip was not known to hold.
For the particular case of filter models one may prove the fip as a simple consequence of the closure of filter models under the contruction
 of direct products, a result that does not appear in the literature and we do not sketch here.

\subsection{Ultraproduct property for i-models}\label{upp-imodels}

In this subsection we deal with the ultraproduct property for the class of i-models: for every non-empty family $\{\bA_i\}_{i \in I}$ of i-webs and 
  every ultrafilter $U$ on $\cP(I)$ the ultraproduct $(\prod_{i\in I}\bA_i^+)/U$ can be embedded into an i-model.
 
Let $J$ be a non-empty set and let $\{\cA_j\}_{j\in J}$ be a family of information systems and let $U$ be a proper ultrafilter on $\cP(J)$. Define
 a binary relation $\theta_U$ on $\prod_{j\in J}A_j$ by setting $(\alpha,\beta) \in \theta_U \iff \{j \in J \st \alpha(j) = \beta(j)\} \in U$. Note
 that $\theta_U$ is an equivalence relation on $\prod_{j\in J}A_j$; we write
 $(\prod_{j\in J}A_j)/U$ for the quotient of $\prod_{j\in J}A_j$ by $\theta_U$.
 As a matter of notation, for every $\ga \in \prod_{j\in J}A_j$ we let 
 $\ga/U = \{\gb \in \prod_{j\in J}A_j \st (\ga,\gb) \in \theta_U\}$ and for every finite subset $a \subseteq_\rf \prod_{j\in J}A_j$, we let
 $a/U = \{\ga/U \st \alpha \in a\}$, i.e., $a/U$ is the finite subset of
 $(\prod_{j\in J}A_j)/U$ constituted by the $\theta_U$-equivalence classes of the tokens of $a$. Since each element $\ga \in a$ is a $J$-indexed
 sequence, we denote by $\ga(j)$ the $j$-th projection of $\ga$ and we let $a(j) = \{\ga(j) \st \ga \in a\}$.

\begin{definition}\label{def:ultra-i-web}
We define an information system $\cP_U = (P_U,\Con_U,\vdash_U,\nu_U)$ as follows:
\[
\begin{array}{l}
P_U = (\prod_{j\in J}A_j)/U \qquad\qquad\qquad\qquad\qquad\qquad\qquad\qquad\qquad\qquad\qquad\qquad \\
\nu_U = (\metal j.\nu_{A_j})/U  \\
a/U \in \Con_U            \quad \text{ iff } \quad \{j \in J \st a(j) \in \Con_{A_j} \} \in U \\
a/U \vdash_U \ga/U \quad\ \ \ \text{ iff } \quad \{j \in J \st a(j) \vdash_{A_j} \ga(j) \} \in U
\end{array}
\] 
We also define an i-web $\bP_U = (\cP_U,\phi_{P_U})$ by setting $\phi_{P_U}(a/U,\alpha/U) = (\metal j.\phi_{A_j}(a(j),\alpha(j)))/U$.
\end{definition}

We leave to the reader the easy verification of the fact that $\cP_U$ and $\bP_U$ indeed are an information system and an i-web, respectively.

We conclude the second main theorem of the section, the one that deals with the ultraproduct property. Let $\{\bA_j\}_{j\in J}$ be a family of i-webs, let $U$ be an ultrafilter over $\cP(J)$ and let $\bP_U$ be the i-web of Definition \ref{def:ultra-i-web}.
 Since $\bP_U$ is an i-web, then $\bP_U^+$ is a reflexive Scott domain and hence a $\lambda$-model. On the other hand each i-web $\bA_j$ gives rise to a reflexive
 Scott domain $\bA_j^+$, which is a $\lambda$-model. Then $(\prod_{j\in J}\bA_j^+)/U$ is an ultraproduct of $\lambda$-models, and thus again a
$\lambda$-model. 

\begin{theorem}\label{thm:embedding}
There exists an embedding of combinatory algebras from the $\lambda$-model $(\prod_{j\in J}\bA_j^+)/U$ into the $\lambda$-model $\bP_U^+$.
\end{theorem}

\begin{proof}
The proof is rather technical and cumbersome. For this reason we state and prove a particular case that only deals with graph models.
We let $x,y,\ldots$ range over elements of $\prod_{j\in J}\bA_j^+$, so that $x(j) \in \bA_j^+$ is a point of the graph model $\bA_j$.
 We write $x/U$ for the equivalence class of $x$ w.r.t. the congruence on $\prod_{j\in J}\bA_j^+$ given by $x \sim_U y \iff \{j\in J\st x(j)=y(j)\}\in U$,
 i.e., $x/U = \{y \in \prod_{j\in J}\bA_j^+ \st x \sim_U y \}$.

Recall that $\sim_U$ is the relation on $\prod_{j\in J}A_j$ given by $\alpha \sim_U \beta  \iff \{j \in J \st \alpha(j) = \beta(j)\}\in U$.
 We define a map $f: (\prod_{j\in J} \bA_j^+)/U \to \bP_U^+$ as follows:

$$ f(x/U) = \{\ga/U \st \ga \in \prod_{j\in J}A_j,\ \forall j\in J.\ \ga(j) \in x(j)\} $$

It is easy to show that the definition of $f$ is independent of the choice of the representatives of $\sim_U$-equivalence classes as, for all $y \in x/U$,
 we have $\{j \in J \st y(j) = x(j) \} \in U$.

We prove that $f$ is injective. Suppose $x/U \neq y/U$ and let $Z = \{j \in J \st x(j) = y(j)\}$.
 Define $X = \{ k \in J \st x(k) \subseteq y(k) \}$ and $Y = \{ k\in J \st y(k) \subseteq x(k) \}$. Then $X \cap Y = Z \not\in U$.
 This means that it is not possible that both $X$ and $Y$ belong to the ultrafilter $U$.
 Assume that $X \not\in U$. Then for every $k\in J-X$ we have $x(k)\not\subseteq y(k)$, so that for each $k\in J-X$ there exists an element $\gamma_k \in A_k$
 such that $\gamma_k \in x(k)-y(k)$. Let $\gd \in \prod_{j\in J}A_j$ be an arbitrary sequence and let $\gb \in \prod_{j\in J}A_j$ be defined by
 $\gb(i) = \gamma_i$ for $i\in J-X$ and $\gb(i) = \gd(i)$ for $i\not\in J-X$. By definition of $f$ we have $\gb/U \in f(x/U)$, while
 $\gb/U \not\in f(y/U)$, so that $f(x/U) \neq f(y/U)$.

Now we prove that $f$ is homomorphism of combinatory algebras. We start proving that $f$ preserves application. We have
\[
\begin{array}{lcl}
f(x/U)\cdot f(y/U) & = & \{\ga/U \st \exists a/U \subseteq_\rf f(y/U).\ \phi_{P_U}(a/U,\ga/U)\in f(x/U) \} \\
  & = & \{\ga/U \st \exists a \subseteq_\rf \prod_{j\in J}A_j.\forall \gc \in a.\forall j\in J.\ \gc(j) \in y(j) \text{ and } \\
  &   & \forall i\in J.\ \phi_{A_i}(a(i),\alpha(i)) \in x(i) \} \\
  & = & \{\ga/U \st \forall j\in J.\exists a \subseteq_\rf y(j).\ \phi_{A_j}(a,\ga(j))\in x(j) \} \\ 
  & = & \{\ga/U \st \forall j\in J.\ \ga(j)\in \{\beta \in A_j \st \exists a \subseteq_\rf y(j).\ \phi_{A_j}(a,\beta)\in x(j)\}\} \\
  & = & \{\ga/U \st \forall j\in J.\ \ga(j)\in x(j)\cdot y(j)\} \\
  & = & f((x\cdot y)/U) \\
  & = & f(x/U \cdot y/U)
\end{array}
\]
We now regard the basic combinators. Recall that by definition for each $j \in J$ we have
 $\ssk^{\bA_j^+}= \Int{\lambda xy.x}^{\bA_j^+}= \{\phi_{A_j}(a,\phi_{A_j}(b,\gb)) \st \gb \in a \}$.
Then
\[ 
\begin{array}{lcl}
f(\ssk^{(\prod_{j\in J}\bA_j^+)/U}) & = & f((\ssk^{\prod_{j\in J}\bA_j^+})/U) \qquad\qquad\qquad\qquad\qquad\qquad\qquad\qquad\qquad\qquad\qquad\qquad\qquad \\
 & = & \{\ga/U \st \ga \in \prod_{j\in J}A_j,\ \forall j\in J.\ \ga(j) \in \ssk^{\bA_j^+}\} \\
 & = & \{\phi_{P_U}(a/\theta_U,\phi_{P_U}(b/U,\gb/U)) \st \gb/U \in a/U \} \\
 & = & \Int{\lambda xy.x}^{\bP_U^+} \\
 & = & \ssk^{\bP_U^+}
 \end{array}
\]
Similarly $f(\sss^{(\prod_{j\in J}\bA_j^+)/U}) = \sss^{\bP_U^+}$.
\end{proof}

We remark that in the general case in which all the $\bA_j$ ($j\in J$) and $\bP_U$ are i-webs the map \\
 $f: (\prod_{j\in J} \bA_j^+)/U \to \bP_U^+$ is defined as
 $f(x/U) = \cldn{\{\ga/U \st \ga \in \prod_{j\in J}A_j,\ \forall j\in J.\ \ga(j) \in x(j)\}}_{P_U}$.

We remarked at the end of Section \ref{upp-imodels} that the fip can be derived for subclasses by suitably modifying the general construction
 detailed for i-models. Also the upp holds for the various classes of models. Here we proved it for graph models, because it looks it looks very
 clear for this case, but the proof can be adapted (adding details and complication) to the other cases.

Summing up, graph models, pcs-models, Krivine models, filter models and in general i-models have both the fip and the upp. For this reason 
 Theorem \ref{thm:main-general} applies to all these classes, producing a minimal model in each case. It is known that there exist filter-easy terms
 \cite{Ale01} as well as graph-easy terms \cite{BaeBoe79} (for example $(\lambda x.xx)(\lambda x.xx)$), and every graph-easy term is also pcs-easy
 and Krivine-easy, since the latter classes contain the graph models. Therefore Theorem \ref{compact} and Theorem \ref{easy} both hold for all these
 classes, saying that each one of them induces a continuum of lambda-theories.

\section{Conclusions}\label{conclusions}

We have presented a method for proving that a given class of models of the $\lambda$-calculus has a minimal element, i.e., an element whose
 $\lambda$-theory is the intersection of all the $\lambda$-theories represented in the class. We have applied this method to the class of i-models, a subclass of Scott models defined in \cite{CarSal09}, containing several well-known
 instances of ``webbed'' models like the graph-models and the filter models living in the category of Scott domains.

Various extensions of this work can be explored, both toward the proof that the whole class of Scott models has the minimality property, and more
 generally toward the application of the method to other classes of models of the $\lambda$-calculus.

Concerning the former extension, a preliminary result would be the finite intersection property for the whole class of Scott models, the completion
 method described in Section \ref{applications} being adapted to i-models.

More generally, it is interesting to notice that webs, even beyond i-webs, are first-order axiomatisable, hence closed by ultraproducts (by the way,
 this observation is an alternative way of showing that Definition \ref{def:ultra-i-web} is sound). 
By providing a first-order axiomatisation of sentences like $\bA^+\vDash M \neq N$, for given terms $M,M$ and web $\bA$, we could invoke Lo\'{s}
 theorem for showing that $(\prod_{j\in J}\bA_j)/U)^+$ and  $(\prod_{j\in J}\bA_j^+)/U$ have the same theory, and hence for deriving a strong
 form of the ultraproduct property for the class of models corresponding to the considered webs. 
 
We conclude this section by providing an outline of a first-order axiomatisation of reflexive information systems.
Let $\cA= (A,\mathrm{Con}_A, \vdash_A, \nu_A)$  be an  information system.  $\cA$ can be defined as a first-order structure as follows: 
for every $n\geq 1$, let $C_n$ be an $n$-ary predicate and $R_{n+1}$ be an $(n+1)$-ary predicate whose intended meanings are:
$$C_n(\ga_1,\dots, \ga_n) \leftrightarrow \{\ga_1,\dots,\ga_n\}\in Con_\cA.$$
and
$$R_{n+1}(\ga_1,\dots, \ga_n,\gb) \leftrightarrow \{\ga_1,\dots,\ga_n\} \vdash_\cA \gb.$$
Then, it is very easy to axiomatise information systems as universal Horn formulas:
\begin{enumerate}
\item $\forall \ga. C_1(\ga)$;
\item $\forall \ga_1\dots\ga_n. C_n(\ga_1,\dots, \ga_n) \to C_k(\ga_{i_1},\dots, \ga_{i_k})$ if $k\leq n$ and $1\leq i_j\leq n$;
\item $\forall \ga_1\dots\ga_n\gb. R_{n+1}(\ga_1,\dots, \ga_n,\gb) \to C_{n+1}(\ga_1,\dots, \ga_n,\gb)$;
\item $\forall \ga_1\dots\ga_n\gb. R_{n+1}(\ga_1,\dots, \ga_n,\gb) \to R_{n+1}(\ga_{\gs(1)},\dots, \ga_{\gs(n)},\gb)$, for every permutation $\gs$;
\item $\forall \ga_1\dots\ga_n\gb_1\dots\gb_k\gamma. (\bigwedge_{1\leq i\leq k} R_{n+1}(\ga_1,\dots, \ga_n,\gb_i))\wedge R_{k+1}(\gb_1,\dots, \gb_k,\gamma) \to R_{n+1}(\ga_1,\dots, \ga_n,\gamma)$;
\item $\forall \ga_1\dots\ga_n. C_n(\ga_1,\dots, \ga_n) \to R_{n+1}(\ga_1,\dots, \ga_n,\ga_i)$;
\item $R_1(\gn)$, for a constant $\gn$.
\end{enumerate}
In a similar but more complicated way it is possible to find a first-order axiomatisation of what is an exponent and a reflexive object in the category $\Inf$. Thus, an untraproduct of reflexive information systems is again a reflexive information system. It deserves to be studied how first-order closure properties of information systems can be transferred to the category $\SD$ of Scott domains.


\bibliographystyle{eptcs}
\bibliography{bibfile}

\end{document}